\documentclass[10pt, a4paper,american]{article}
\usepackage[pdftex]{graphicx}
\usepackage{amsmath, amssymb, amsthm, amsfonts}
\usepackage{float, color}
\usepackage{ascmac, fancybox}
\usepackage{bm}

\usepackage{bbm}

\usepackage[margin=1in]{geometry}

\makeatletter

\theoremstyle{plain}
\numberwithin{equation}{section}
\theoremstyle{plain}
\newtheorem{thm}{Theorem}[section]
\theoremstyle{plain}

\theoremstyle{definition}
\newtheorem{defn}[thm]{Definition}
\theoremstyle{definition}

\theoremstyle{plain}

\theoremstyle{plain}
\newtheorem{cor}[thm]{Corollary}
\theoremstyle{definition}

\makeatother

\usepackage{bm}

\usepackage[utf8]{inputenc}
\usepackage[T1]{fontenc}
\usepackage{mathptmx}

\usepackage{amsmath}
\usepackage{amsthm}
\usepackage{amssymb}
\usepackage{stmaryrd}

\usepackage{stmaryrd}
\usepackage{cases}
\usepackage{color}
\usepackage{comment}

\begin{document}
\title{Quantum monotone metrics induced from trace non-increasing maps and additive noise}
\author{
	Koichi Yamagata%
	\thanks{koichi.yamagata@uec.ac.jp}\\
	{The University of Electro-Communications Department of Informatics,}\\
		{1-5-1, Chofugaoka, Chofu, Tokyo 182-8585, Japan} 
}%
\date{}

\maketitle

\begin{abstract}
Quantum monotone metric was introduced by Petz, and it was proved
that quantum monotone metrics on the set of quantum states with trace
one were characterized by operator monotone functions. Later, these
were extended to monotone metrics on the set of positive operators
whose traces are not always one based on completely positive, trace
preserving (CPTP) maps. 
It was shown that these extended monotone metrics were characterized by operator monotone functions continuously
parameterized by traces of 
positive operators,
and did not have some ideal properties
such as monotonicity and convexity with respect to 
the positive operators.
In this paper, we introduce another extension of quantum monotone metrics
which have monotonicity under completely positive, trace non-increasing
(CPTNI) maps and additive noise. We prove that our extended monotone
metrics can be characterized only by static operator monotone functions
from few assumptions without assuming continuities of metrics. 
We show that our monotone metrics have some natural properties such as
additivity of direct sum, convexity and monotonicity with respect
to 
positive operators.
\end{abstract}

\maketitle

\newcommand{\argmax}{\mathop{\rm arg~max}\limits}
\newcommand{\argmin}{\mathop{\rm arg~min}\limits}

\global\long\def\E{\mathcal{E}}%
\global\long\def\S{\mathcal{S}}%
\global\long\def\R{\mathbb{R}}%
\global\long\def\C{\mathbb{C}}%
\global\long\def\N{\mathbb{N}}%
\global\long\def\D{\mathcal{D}}%
\global\long\def\M{\mathcal{M}}%
\global\long\def\X{\mathcal{X}}%
\global\long\def\F{\mathcal{F}}%
\global\long\def\B{\mathcal{B}}%
\global\long\def\T{\mathcal{T}}%
\global\long\def\para{\xi}%
\global\long\def\P{\mathcal{P}}%
\global\long\def\Para{\Xi}%
\global\long\def\H{\mathcal{H}}%
\global\long\def\Tr{{\rm Tr}\,}%
\global\long\def\Ch{\mathcal{C}}%
\global\long\def\bR{\mathbf{R}}%
\global\long\def\bL{\mathbf{L}}%
\global\long\def\ket#1{\left|#1\right\rangle }%
\global\long\def\bra#1{\left\langle #1\right|}%
\global\long\def\braket#1#2{\left\langle #1\mid#2\right\rangle }%

\section{Introduction}

In classical statistics, Cencov \cite{cencov} proved that the Fisher
information metric is the only Riemannian metric on families of probabilities
up to rescaling that has monotonicity under Markov maps. Petz extended
Cencov's theorem to families of quantum states, and it was revealed
that there is a one-to-one correspondence between quantum monotone
metrics and operator monotone functions \cite{petz}. To introduce
Petz's characterization of monotone metrics,
let us define some notations.

Let $\B(\C^{n})$ be the set of all linear operators on $\C^{n}$, and let 
$\S^{++}(\C^{n})=\left\{ \rho\in\B(\C^{n})\mid\Tr\rho=1,\,\rho>0\right\} $
be the set of all strictly positive quantum states on $\C^{n}$. The
quantum analog of Markov map can be expressed as completely positive,
trace preserving map (CPTP map) $T:\B(\C^{n})\rightarrow\B(\C^{m})$,
and it can be represented by operator sum representation
\[
T(X)=\sum_{i=1}^{k}A_{i}XA_{i}^{*}
\]
with linear maps $A_{1},\dots,A_{k}$ from $\C^{n}$ to $\C^{m}$
such that
\[
\sum_{i=1}^{k}A_{i}^{*}A_{i}=I.
\]
We denote by $\Ch(\C^{n},\C^{m})$ the set of all CPTP maps from $\B(\C^{n})$
to $\B(\C^{m})$. Let $\bL_{\rho}$ and $\bR_{\rho}$ be super operators
on $\B(\C^{n})$ defined by
\begin{align*}
\bL_{\rho}(X) & =\rho X,\\
\bR_{\rho}(X) & =X\rho,
\end{align*}
with a strictly positive operator $\rho$. 

Petz defined quantum monotone metric as follows.
\begin{defn}
	\label{def:petz}
	A family of functions $\{ K_{\centerdot}^{(n)}(\cdot,\cdot) \}_{n\in\N^+}$ 
	from $\S^{++}(\C^n) \times \B(\C^n) \times \B(\C^n)$ to $\C$
	is a family of monotone metrics, if 
	 the following conditions hold:
	\begin{description}
		\item [{(a)}] For every $n\in\N^{+}$ and $\rho\in\S^{++}(\C^n)$ the map 
		\[
		K^{(n)}_{\rho}: \B(\C^n)\times\B(\C^n) \to \C 
		\qquad (X,Y)\mapsto K^{(n)}_{\rho}(X,Y) 
		\] 
		is an inner product.
		\item [{(b)}] 
		For every $n,m\in\N^+$,  
		$X\in\B(\C^{n})$, 
		CPTP map $T\in\Ch(\C^{n},\C^{m})$,
		$\rho\in\S^{++}(\C^{n})$ 
		such that $T(\rho)\in\S^{++}(\C^{m})$, 
		the inequality
		\[
		K^{(m)}_{T(\rho)}\left(T(X),T(X)\right)\leq K^{(n)}_{\rho}\left(X,X\right)
		\]
		holds.
		\item [{(c)}] 
		For every  $n\in\N^+$ and $X\in\B(\C^{n})$,
		the map $\rho\mapsto K^{(n)}_{\rho}\left(X,X\right)$ is continuous.
	\end{description}
\end{defn}

Any monotone metric can be characterized by operator monotone functions
as follows. (See Appendix \ref{sec:ope_monotone} for a brief account of operator monotone functions. 
See \cite{Bhatia} for more details.)
\begin{thm}[Petz\cite{petz}]
	$\{ K_{\centerdot}^{(n)}(\cdot,\cdot) \}_{n\in\N^+} $ 
	is 
	a family of monotone metrics if and only if there
	exists an operator monotone function $f:(0,\infty)\rightarrow(0,\infty)$
	and a non-negative constant $c\in\R$ such that
	\begin{equation}
	K^{(n)}_{\rho}(X,Y)=\Tr X^{*}[(\bR_{\rho}f(\bL_{\rho}\bR_{\rho}^{-1}))^{-1}Y]+c(\Tr X)^{*}(\Tr Y) \label{eq:petz}
	\end{equation}
	for $n\in\N^+$, $\rho\in \S^{++}(\C^n)$, and $X,Y\in\B(\C^n)$.
\end{thm}

Note that the second term in (\ref{eq:petz}) does not appear when
$X$ and $Y$ are derivatives of parameterized density operators whose
traces are fixed to one, and $\Tr X=\Tr Y=0$.  However, if we consider
a extended parametric family containing states whose traces are not
fixed to one, the traces of derivatives are not always zero.

In a previous study, monotone metrics on the quantum state family
$\S^{++}(\C^{n})$ are extended to strictly positive operators $\B^{++}(\C^{n})$
in Ref. \cite{kumagai}. 
In their study, 
a family of functions 
$\{ K^{(n)}_{\centerdot}(\cdot,\cdot) \}_{n\in\N^+}$
from $\B^{++}(\C^n) \times \B(\C^n) \times \B(\C^n)$ to $\C$
is a family of monotone metrics if following conditions hold:
\begin{description}
	\item [{(a)}] For every $n\in\N^{+}$ and $\rho\in\B^{++}(\C^n)$ the map 
	\[
	K^{(n)}_{\rho}: \B(\C^n)\times\B(\C^n) \to \C 
	\qquad (X,Y)\mapsto K^{(n)}_{\rho}(X,Y) 
	\] 
	is an inner product.
	\item [{(b)}] 
	For every $n,m\in\N^+$,  
	$X\in\B(\C^{n})$, 
	CPTP map $T\in\Ch(\C^{n},\C^{m})$,
	$\rho\in\B^{++}(\C^{n})$ 
	such that $T(\rho)\in\B^{++}(\C^{m})$, 
	the inequality
	\[
	K^{(m)}_{T(\rho)}\left(T(X),T(X)\right)\leq K^{(n)}_{\rho}\left(X,X\right)
	\]
	holds.
	\item [{(c)}] 
	For every  $n\in\N^+$ and $X\in\B(\C^{n})$,
	the map $\rho\mapsto K^{(n)}_{\rho}\left(X,X\right)$ is continuous.
\end{description}
These conditions are same as Definition \ref{def:petz} except that
$\S^{++}(\C^{n})$ is replaced by $\B^{++}(\C^{n})$. They proved
that $\{ K^{(n)}_{\centerdot}(\cdot,\cdot) \}_{n\in\N^+}$ is a family of monotone metrics if and only if
there exist a continuous function $b:\R^{++}\rightarrow\R$ and a
continuous family of operator monotone functions $\{f_{t}:\R^{++}\rightarrow\R^{++}\}_{t\in\R^{++}}$
such that
\begin{equation}
K^{(n)}_{\rho}(X,Y)=\Tr X^{*}[(\bR_{\rho}f_{\Tr\rho}(\bL_{\rho}\bR_{\rho}^{-1}))^{-1}Y]+b(\Tr\rho)(\Tr X)^{*}(\Tr Y)\label{eq:kuma}
\end{equation}
with $f_{t}(1)^{-1}+t\,b(t)>0$,
for $n\in\N^+$, $\rho\in\B^{++}(\C^n)$, and $X,Y\in\B(\C^n)$.  
To distinguish this metric from
ours defined later, we call it CPTP monotone metric. It can be seen
that this CPTP monotone  metric does not have some desirable properties.
For example, $K^{(n)}_{\rho}(X,X)$ is not convex with respect to $\rho$
in general. Further, it does not have the additivity with respect
to direct sum:
\[
K^{(n_1+n_2)}_{\rho_{1}\oplus\rho_{2}}
\left(X_{1}\oplus X_{2},X_{1}\oplus X_{2}\right)
=K^{(n_1)}_{\rho_{1}} \left(X_{1},X_{1}\right)
+K^{(n_2)}_{\rho_{2}} \left(X_{2},X_{2}\right).
\]
This means that the inner product structures are different between
whole and part. 

In this study, we introduce another extension of quantum monotone
metrics which have monotonicity under completely positive, trace non-increasing
maps. We prove that our extended monotone metrics can be characterized
by fixed operator monotone functions from few assumptions without
assuming continuities of metrics. We show that our monotone metrics
have some natural properties such as additivity of direct sum, convexity
and monotonicity with respect to unnormalized states.

\section{Quantum monotone metrics induced from CPTNI maps and additive noise}

In quantum mechanics, a quantum operation $T$ is used to describe
transformations of quantum states, and it must satisfy $0\leq\Tr T(\rho)\leq1$
for any quantum state $\rho$ to be physical\cite{nielsen,caves}. 
A quantum operations
can be expressed as a completely positive, trace non-increasing (CPTNI)
maps $T:\B(\C^{n})\rightarrow\B(\C^{m})$, and it can be represented
by operator sum representation
\[
T(X)=\sum_{i=1}^{k}A_{i}XA_{i}^{*}
\]
with linear maps $A_{1},\dots,A_{k}$ from $\C^{n}$ to $\C^{m}$
such that
\[
\sum_{i=1}^{k}A_{i}^{*}A_{i}\leq I.
\]
If there is a state $\rho$ such that $\Tr [T(\rho)] < 1$, then the quantum operation $T$ 
does not provide a complete description of processes that may occur in a system, and
$\Tr[T(\rho)]$ is equal to the probability 
that $T$ occurs\cite{nielsen}.
In this sense, unnormalized states having traces less than one
can be interpreted  as 
results of incomplete quantum operations. 
We denote by $\overline{\Ch}(\C^{n},\C^{m})$ the set of all CPTNI maps
from $\B(\C^{n})$ to $\B(\C^{m})$. We denote by
$\overline{\S}^{++}(\C^{n}):=\left\{ \rho\in\B(\C^n)\mid\Tr\rho\leq1,\,\rho>0\right\} $
and $\overline{\S}(\C^{n}):=\left\{ \rho\in\B(\C^n)\mid\Tr\rho\leq1,\,\rho\geq0\right\} $
the set of all strictly and non-strictly positive operators with traces
less than one. To extend monotone metric, it is natural to consider
a condition $K_{T(\rho)}\left(T(X),T(X)\right)\leq K_{\rho}\left(X,X\right)$
for every CPTNI map $T$. 

When the trace of  an unnormalized quantum state $\rho$ is less than one, $\rho+\sigma$
may also physical with $\sigma\geq0$ such that $\Tr(\rho+\sigma)\leq1$.
In this case, $\sigma$ is considered to be noise for $\rho$. Therefore
a metric on $\overline{\S}^{++}(\C^{n})$ which has monotonicity under
noise should satisfy $K_{\rho+\sigma}\left(X,X\right)\leq K_{\rho}\left(X,X\right)$.
In this inequality, $X$ does not need to be changed under this kind
of noise because the derivative of $\rho+tX+\sigma$ with respect
to $t$ is $X$. Note that this kind of noise is not necessary to
be considered when only normalized states and CPTP maps are treated
because $\rho\mapsto T(\rho)+(\Tr\rho-\Tr T(\rho))\sigma$ is a CPTP
map for $\sigma\in\S^{++}(\C^{m})$ and $T\in\overline{\Ch}(\C^{n},\C^{m})$,
in fact, it has an operator sum representation
\[
\sum_{i=1}^{k}A_{i}\rho A_{i}^{*}+\sum_{st}B_{st}\rho B_{st}^{*}
\]
where
\[
B_{st}=\sqrt{\lambda_{s}}\text{\ensuremath{\ket{e_{s}}\bra{e_{t}}}}\sqrt{I-\sum_{i=1}^{k}A_{i}^{*}A_{i}},
\]
with the spectral decomposition $\sigma=\sum_{s}\lambda_{s}\ket{e_{s}}\bra{e_{s}}$
and the operator sum representation $T(X)=\sum_{i=1}^{k}A_{i}XA_{i}^{*}$.

Based on the above considerations, we define quantum monotone metrics
which have monotonicity under CPTNI maps as follows.
\begin{defn}
\label{def:monotone_cptni}
A family of functions $\{ K_{\centerdot}^{(n)}(\cdot,\cdot) \}_{n\in\N^+}$ 
from $\overline{\S}^{++}(\C^n) \times \B(\C^n) \times \B(\C^n)$ to $\C$
is a family of monotone metrics, if 
the following conditions hold:
\begin{description}
	\item [{(a)}] For every $n\in\N^{+}$ and $\rho\in\overline{\S}^{++}(\C^n)$ the map 
	\[
	K^{(n)}_{\rho}: \B(\C^n)\times\B(\C^n) \to \C 
	\qquad (X,Y)\mapsto K^{(n)}_{\rho}(X,Y) 
	\] 
	is an inner product.
	\item [{(b)}] 
	For every $n,m\in\N^+$,  CPTNI map $T\in\overline{\Ch}(\C^{n},\C^{m})$, $X\in\B(\C^{n})$,
	$\rho\in\overline{\S}^{++}(\C^{n})$, and $\sigma\in\overline{\S}(\C^{m})$
	such that $T(\rho)+\sigma\in\overline{\S}^{++}(\C^{m})$,
	the inequality
	\begin{equation}\label{eq:cptni_mono}
	K^{(m)}_{T(\rho)+\sigma}\left(T(X),T(X)\right)\leq K^{(n)}_{\rho}\left(X,X\right)
	\end{equation}
	holds.
\end{description}
\end{defn}

We call these metrics CPTNI monotone metrics to distinguish them from
metrics based on CPTP maps (\ref{eq:kuma}). We prove the following
Theorem.
\begin{thm}
$\{ K^{(n)}_{\centerdot}(\cdot,\cdot)\}_{n\in\N^+}$ is a family of CPTNI monotone metrics 
if and only if there exists an operator monotone function $f:(0,\infty)\rightarrow(0,\infty)$
such that
\begin{equation}
K^{(n)}_{\rho}(X,Y)=\Tr X^{*}[(\bR_{\rho}f(\bL_{\rho}\bR_{\rho}^{-1}))^{-1}Y]\label{eq:petz2}
\end{equation}
for $n\in\N^+$, $\rho\in\overline{\S}^{++}(\C^n)$, and $X,Y\in\B(\C^n)$.
\end{thm}

Note that the continuous condition is not necessary to characterize
CPTNI monotone metrics unlike Definition \ref{def:petz}. Further,
the second term in RHS of (\ref{eq:petz}) does not appear in this
theorem, and the operator monotone $f$ does not depend on $\Tr\rho$
unlike CPTP monotone metrics (\ref{eq:kuma}). 
\begin{proof}[The proof of ``only if'' part]
The inequality \eqref{eq:cptni_mono} implies the unitary covariance 
\[
K^{(n)}_{\rho}(X,X)=K^{(n)}_{U\rho U^{*}}(UXU^{*},UXU^{*})
\]
 for any unitary operator $U$ because $U$ is an invertible CPTNI map, and 
 \[
 K^{(n)}_{\rho}(X,Y)=K^{(n)}_{U\rho U^{*}}(UXU^{*},UYU^{*})
 \]
 due to the polarization identity of the inner product. 
So we can assume $\rho={\rm Diag}(p_{1},\dots,p_{n})$ is diagonal, 
 since we can freely choose such basis for further computations which consists of the
 eigenvectors 
 of $\rho$.
 We denote by $E_{ij}^{(n)}$ the matrix
unit in $\B(\C^{n})$. 
We characterize all elements $K^{(n)}_{\rho}(E_{ij}^{(n)},E_{kl}^{(n)})$
of $K^{(n)}_{\rho}$ for $i,j,k,l\in\{1,\dots,n\}$ by characterizing the
following four types of elements:
\begin{description}
\item [{(i)}] $K^{(n)}_{\rho}(E_{12}^{(n)},E_{kl}^{(n)})$ ($k\not=1$)
\item [{(ii)}] $K^{(n)}_{\rho}(E_{11}^{(n)},E_{22}^{(n)})$
\item [{(iii)}] $K^{(n)}_{\rho}(E_{12}^{(n)},E_{12}^{(n)})$
\item [{(iv)}] $K^{(n)}_{\rho}(E_{11}^{(n)},E_{11}^{(n)})$
\end{description}
Other elements of $K^{(n)}_{\rho}$ can be obtained by replacing bases of
the Hilbert space $\C^{n}$.

(i) Let $U={\rm Diag(c,1,\dots,1)}\in\B(\C^{n})$ be an unitary operator
with $|c|=1$. When $k\not=1$, we have
\[
K^{(n)}_{U\rho U^{*}}(UE_{12}^{(n)}U^{*},UE_{kl}^{(n)}U^{*})=\begin{cases}
\bar{c}^{2}K^{(n)}_{\rho}(E_{12}^{(n)},E_{kl}^{(n)}) & \text{if }l=1,\\
\bar{c}K^{(n)}_{\rho}(E_{12}^{(n)},E_{kl}^{(n)}) & \text{if }l\not=1.
\end{cases}
\]
Therefore $K^{(n)}_{\rho}(E_{12}^{(n)},E_{kl}^{(n)})=0$. 

(ii) Let $T\in\overline{\Ch}(\C^{n},\C^{n})$ be a CPTNI map defined
by
\[
T(E_{ij}^{(n)})=\begin{cases}
E_{11}^{(n)} & \text{if }(i,j)=(1,1),\\
0 & \text{if }(i,j)\not=(1,1),
\end{cases}
\]
and let $\sigma={\rm Diag}(0,\lambda_{2},\dots,\lambda_{n})\in\overline{\S}(\C^{n})$.
By the definition of CPTNI monotone metrics, for any $\lambda\in\R$,
\begin{align*}
K^{(n)}_{\rho}(E_{11}^{(n)}+\lambda E_{22}^{(n)},E_{11}^{(n)}+\lambda E_{22}^{(n)}) & \geq K^{(n)}_{T(\rho)+\sigma}(T(E_{11}^{(n)}+\lambda E_{22}^{(n)}),T(E_{11}^{(n)}+\lambda E_{22}^{(n)}))\\
 & =K^{(n)}_{\rho}(E^{(n)}_{11},E^{(n)}_{11}).
\end{align*}
This means $E_{11}^{(n)}$ and $E_{22}^{(n)}$ are orthogonal with
respect to the inner product $K^{(n)}_{\rho}$, that is, 
\[
K^{(n)}_{\rho}(E_{11}^{(n)},E_{22}^{(n)})=0.
\]
Note that the second terms of RHSs of (\ref{eq:petz}) and (\ref{eq:kuma})
are vanished here by using CPTNI maps and additive noise. 

(iii) By using a CPTNI map $T_{1}\in\overline{\Ch}(\C^{n},\C^{2})$ defined
by
\[
T_{1}(E_{ij}^{(n)})=\begin{cases}
E_{ij}^{(2)} & \text{if }i,j\leq2,\\
O & \text{otherwise},
\end{cases}
\]
we have
\[
K^{(n)}_{\rho}(E_{12}^{(n)},E_{12}^{(n)})\ge K^{(n)}_{T_{1}(\rho)}(T_{1}(E_{12}^{(n)}),T_{1}(E_{12}^{(n)}))=K^{(2)}_{{\rm {\rm Diag}}\left(p_{1},p_{2}\right)}(E_{12}^{(2)},E_{12}^{(2)}).
\]
By using a CPTNI map $T_{2}\in\overline{\Ch}(\C^{2},\C^{n})$ defined
by 
\[
T_{2}(E_{ij}^{(2)})=E_{ij}^{(n)}
\]
 and $\sigma={\rm Diag}(0,0,\lambda_{3},\dots,\lambda_{n})\in\overline{\S}(\C^{n})$,
we have
\[
K^{(2)}_{{\rm {\rm Diag}}\left(p_{1},p_{2}\right)}(E_{12}^{(2)},E_{12}^{(2)})\ge K^{(2)}_{T_{2}({\rm {\rm Diag}}\left(p_{1},p_{2}\right))+\sigma}(T_{2}(E_{12}^{(2)}),T_{2}(E_{12}^{(2)}))=K^{(n)}_{\rho}(E_{12}^{(n)},E_{12}^{(n)}).
\]
Therefore 
\[
K^{(n)}_{\rho}(E_{12}^{(n)},E_{12}^{(n)})=K^{(2)}_{{\rm {\rm Diag}}\left(p_{1},p_{2}\right)}(E_{12}^{(2)},E_{12}^{(2)})=:g(p_{1},p_{2})
\]
depends only on $p_{1},p_{2}$. Note that this fact is different from
CPTP monotone metrics (\ref{eq:kuma}) which depends on a trace of
$\rho$. 

Let $S_{1}\in\overline{\Ch}(\C^{2},\C^{2}\otimes\C^{m})$ be a CPTNI
map defined by
\begin{equation}
S_{1}(E^{(2)})=E^{(2)}\otimes I^{(m)}/m\label{eq:S1}
\end{equation}
for any $E^{(2)}\in\B(\C^{2})$ with an identity operator $I^{(m)}\in\B(\C^{m})$,
and let $S_{2}\in\overline{\Ch}(\C^{2}\otimes\C^{m},\C^{2})$ be a CPTNI
map defined by
\begin{equation}
S_{2}(E^{(2m)})={\rm Tr}_{2}E^{(2m)}\label{eq:S2}
\end{equation}
for any $E^{(2m)}\in\B(\C^{2}\otimes\C^{m})$ where ${\rm Tr}_{2}$
is a partial trace with respect to $\C^{m}$. Because $S_{2}(S_{1}(E^{(2)}))=E^{(2)}$,
we have
\begin{align*}
g(p_{1},p_{2}) & =K^{(2)}_{{\rm {\rm Diag}}\left(p_{1},p_{2}\right)}(E_{12}^{(2)},E_{12}^{(2)})=K^{(2 m)}_{{\rm {\rm Diag}}\left(p_{1},p_{2}\right)\otimes I^{(m)}/m}(E_{12}^{(2)}\otimes I^{(m)}/m,E_{12}^{(2)}\otimes I^{(m)}/m)\\
 & =\frac{1}{m}g\left(\frac{p_{1}}{m},\frac{p_{2}}{m}\right),
\end{align*}
for any $m\in\mathbb{N}$. From this, we have, for any rational number
$q=\frac{m_{1}}{m_{2}}$ with $m_{1},m_{2}\in\N^+$ such that $0<qp_{1}+qp_{2}\leq1$,
\[
g(qp_{1},qp_{2})=g(\frac{m_{1}}{m_{2}}p_{1},\frac{m_{1}}{m_{2}}p_{2})=\frac{1}{m_{1}}g(\frac{1}{m_{2}}p_{1},\frac{1}{m_{2}}p_{2})=\frac{m_{2}}{m_{1}}g(p_{1},p_{2})=\frac{1}{q}g(p_{1},p_{2}).
\]
Because a function $q\mapsto g(qp_{1},qp_{2})$ monotonously decreases
due to the definition \ref{def:monotone_cptni}, we have
\begin{equation}
g(qp_{1},qp_{2})=\frac{1}{q}g(p_{1},p_{2})\label{eq:rational_g}
\end{equation}
 for any real number $q\in(0,\frac{1}{p_{1}+p_{2}}]$. Note that we
don't require the continuity of $K^{(n)}_{\rho}$ to obtain (\ref{eq:rational_g}). 

We can define a function $f:(0,\infty)\rightarrow(0,\infty)$ such
that 
\[
g(p_{1},p_{2})=\frac{1}{p_{2}f(p_{1}/p_{2})}
\]
because $p_{2}g(p_{1},p_{2})$ depends only on $p_{1}/p_{2}$ due
to (\ref{eq:rational_g}). We prove $f$ is an operator monotone function.
Let
\[
\bar{X}=\begin{pmatrix}0 & 0\\
X & 0
\end{pmatrix}\in\B(\C^{2m})
\]
and
\[
\bar{\rho}=\begin{pmatrix}\epsilon I^{(m)} & 0\\
0 & \rho
\end{pmatrix}\in\overline{\S}^{++}(\C^{2m})
\]
are an observable and an unnormalized state represented by block matrices with a
positive real number $\epsilon$. It follows that
\[
K^{(2 m)}_{\bar{\rho}}\left(\bar{X},\bar{X}\right)=\Tr X^{*}\left[\epsilon f(\frac{\rho}{\epsilon})\right]^{-1}X.
\]
Because of definition \ref{def:monotone_cptni}, $\rho\leq\rho'$
implies $f(\frac{\rho}{\epsilon})\leq f(\frac{\rho'}{\epsilon})$.
Therefore $f$ is an operator monotone function. Note that every operator
monotone function on $(0,\infty)$ is continuous and operator concave.
(See Appendix \ref{sec:ope_monotone}). 

(iv) By a similar discussion as (iii), $K^{(n)}_{\rho}(E_{11}^{(n)},E_{11}^{(n)})=:g(p_{1})$
depends only on $p_{1}$. It follows that
\[
g(p_{1})=K^{(2)}_{\frac{1}{2}p_{1}I^{(2)}}(\frac{1}{2}I^{(2)},\frac{1}{2}I^{(2)})=K^{(2)}_{\frac{1}{2}p_{1}I^{(2)}}(\frac{1}{2}X,\frac{1}{2}X)=\frac{1}{p_{1}f(1)},
\]
with $X=\begin{pmatrix}0 & 1\\
1 & 0
\end{pmatrix}$, because eigenvalues of $X$ are $\pm1$.
\end{proof}
Note that the above proof derives continuity of the metric $K^{(n)}_{\rho}(X,X)$
with respect to $\rho$ from only a few assumptions. Also note that,
unlike (\ref{eq:kuma}), the variety of CPTNI monotone metrics depends
only on an operator monotone function $f$.

\begin{proof}[The proof of ``if'' part]
When $f(x)=1$, we prove a metric defined by
\begin{equation}
K^{(n)}_{\rho}(X,X)=\Tr X^{*}\left[\bR_{\rho}^{-1}X\right]=\Tr X\rho^{-1}X^{*}\label{eq:RLD}
\end{equation}
is a CPTNI monotone metric. Because
\[
\begin{pmatrix}\rho & X^{*}\\
X & X\rho^{-1}X^{*}
\end{pmatrix}\geq0,
\]
it follows
\[
\begin{pmatrix}T(\rho)+\sigma & T(X^{*})\\
T(X) & T(X\rho^{-1}X^{*})
\end{pmatrix}\geq\begin{pmatrix}T(\rho) & T(X^{*})\\
T(X) & T(X\rho^{-1}X^{*})
\end{pmatrix}\geq0
\]
for any CPTNI map $T\in\overline{\Ch}(\C^{n},\C^{m})$ and a positive operator $\sigma$. Then it
follows
\[
T(X)\left(T(\rho)+\sigma\right)^{-1}T(X^{*})\leq T(X\rho^{-1}X^{*})
\]
by considering the Schur complement of the block matrix. Because $T$
is trace non-increasing,
\begin{align*}
K^{(m)}_{T(\rho)+\sigma}(T(X),T(X)) & =\Tr T(X)\left(T(\rho)+\sigma\right)^{-1}T(X^{*})\\
 & \leq\Tr T(X\rho^{-1}X^{*})\leq\Tr X\rho^{-1}X^{*}\\
 & =K^{(n)}_{\rho}(X,X).
\end{align*}
Therefore $(K^{(n)}_{\centerdot}(\cdot,\cdot))_{n\in\N^+}$ is a family of CPTNI monotone metrics. Similarly, when
$f(x)=x$, 
\begin{equation}
K^{(n)}_{\rho}(X,X)=\Tr X^{*}\left[\bL_{\rho}^{-1}X\right]=\Tr X^{*}\rho^{-1}X\label{eq:LLD}
\end{equation}
is also CPTNI operator metric. 

When $f:(0,\infty)\rightarrow(0,\infty)$ is any operator monotone
function, $h(x)=\frac{x}{f(x)}$ is also operator monotone (See Appendix
\ref{sec:ope_monotone}), and $K^{(n)}_{\rho}(X,X)$ can be rewritten to
\begin{equation}
K^{(n)}_{\rho}(X,X)=\Tr X^{*}\left[(\bL_{\rho}^{-1}m_{h}\bR_{\rho}^{-1})X\right],\label{eq:ope_mean}
\end{equation}
where
\[
A\,m_{h}B=A^{1/2}\,h(A^{-1/2}BA^{-1/2})A^{1/2}
\]
is operator mean  of strictly positive operators
$A,B$ on a Hilbert space \cite{kubo_ando}. Operator mean of non-negative operators
$A$ and $B$ is $A\,m_{h}B=\lim_{\epsilon\ssearrow0}(A+\epsilon I)\,m_{h}(B+\epsilon I)$.
It is known that operator means fulfill inequalities
\[
A\,m_{h}B\leq A'\,m_{h}B'
\]
if $A\leq A'$ and $B\leq B'$, and
\[
C(A\,m_{h}B)C^{*}\leq(CAC^{*})\,m_{h}(CBC^{*})
\]
for any operator $C$. (See Appendix \ref{sec:ope_mean} for a brief
account of operator means.) By using them, we have
\begin{align*}
K^{(m)}_{T(\rho)+\sigma}(T(X),T(X)) & =\Tr T(X^{*})\left[(\bL_{T(\rho)+\sigma}^{-1}m_{h}\bR_{T(\rho)+\sigma}^{-1})T(X)\right]\\
 & =\Tr X^{*}\left[T^{*}(\bL_{T(\rho)+\sigma}^{-1}m_{h}\bR_{T(\rho)+\sigma}^{-1})T\right]X\\
 & \leq\Tr X^{*}\left[\left(T^{*}\bL_{T(\rho)+\sigma}^{-1}T\right)m_{h}\left(T^{*}\bR_{T(\rho)+\sigma}^{-1}T\right)\right]X\\
 & \leq\Tr X^{*}\left[\bL_{\rho}^{-1}m_{h}\bR_{\rho}^{-1}\right]X=K^{(n)}_{\rho}(X,X),
\end{align*}
where CPTNI monotonicity of (\ref{eq:RLD}) and (\ref{eq:LLD}) are
used in the last inequality. 
This proves $(K^{(n)}_{\centerdot}(\cdot,\cdot))_{n\in\N^+}$
is a family of CPTNI monotone
metrics for any operator monotone function $f$. 
\end{proof}
\begin{cor}
\label{cor:direct_sum}For any CPTNI monotone metric,
\[
K^{(n_1+n_2)}_{\rho_{1}\oplus\rho_{2}}\left(X_{1}\oplus X_{2},Y_{1}\oplus Y_{2}\right)
=K^{(n_1)}_{\rho_{1}}\left(X_{1},Y_{1}\right)+K^{(n_2)}_{\rho_{2}}\left(X_{2},Y_{2}\right)
\]
with $\rho_{i}\in\overline{\S}^{++}(\C^{n_{i}})$ and $X_{i},Y_{i},\in\B(\C^{n_{i}})$
for $i=1,2$. 
\end{cor}

This is a natural property meaning that the inner product structure
of the whole and part is the same, while CPTP monotone metrics (\ref{eq:kuma})
do not have this property. The following two corollaries about monotonicity
and convexity are also natural consequences which CPTP monotone metrics
do not have. 
\begin{cor}
\label{cor:convexity}For any CPTNI monotone metric,
\[
K^{(n)}_{\frac{\rho_{1}+\rho_{2}}{2}}\left(\frac{X_{1}+X_{2}}{2},\frac{X_{1}+X_{2}}{2}\right)\leq\frac{1}{2}\left\{ K^{(n)}_{\rho_{1}}\left(X_{1},X_{1}\right)+K^{(n)}_{\rho_{2}}\left(X_{2},X_{2}\right)\right\} .
\]
\end{cor}

\begin{proof}
By using the above corollary and a partial trace, for $\rho_{1},\rho_{2}\in\overline{\S}^{++}(\C^{n})$
and $X_{1},X_{2}\in\B(\C^{n})$, 
\begin{align*}
\frac{1}{2}\left\{ K^{(n)}_{\rho_{1}}\left(X_{1},X_{1}\right)+K^{(n)}_{\rho_{2}}\left(X_{2},X_{2}\right)\right\}  & =K^{(n)}_{\frac{\rho_{1}}{2}}\left(\frac{X_{1}}{2},\frac{X_{1}}{2}\right)+K^{(n)}_{\frac{\rho_{2}}{2}}\left(\frac{X_{2}}{2},\frac{X_{2}}{2}\right)\\
 & =K^{(n)}_{\frac{\rho_{1}\oplus\rho_{2}}{2}}\left(\frac{X_{1}\oplus X_{2}}{2},\frac{X_{1}\oplus X_{2}}{2}\right)\\
 & \geq K^{(n)}_{\frac{\rho_{1}+\rho_{2}}{2}}\left(\frac{X_{1}+X_{2}}{2},\frac{X_{1}+X_{2}}{2}\right).
\end{align*}
\end{proof}
\begin{cor}
\label{cor:monotonicity}For any CPTNI monotone metric $K^{(n)}_{\rho}$,
a function $\rho\mapsto K^{(n)}_{\rho}(X,X)$ is monotonically decreasing
and convex with respect to $\rho\in\overline{\S}(\C^{n})$. 
\end{cor}

\begin{proof}
The monotonicity is obvious from the definition of CPTNI monotone
metric. The convexity is also obvious form the above corollary. 
\end{proof}
In this section, we consider CPTNI monotone metric $K^{(n)}_{\rho}$ for
physical unnormalized state $\rho$ which is restricted to $\overline{\S}^{++}(\C^{n})=\left\{ \rho\in\B(\H)\mid\Tr\rho\leq1,\,\rho>0\right\} $.
However, it is easy to generalize $\overline{\S}^{++}(\C^{n})$ to $\overline{\overline{\S}}^{++}(\C^{n}):=\left\{ \rho\in\B(\H)\mid\rho>0\right\} $.
In fact, $\overline{\S}^{++}(\C^{n})$ can be replaced by $\overline{\overline{\S}}^{++}(\C^{n})$
in this section without any restrictions.

\section{Conclusion}

In the present paper, we introduced CPTNI monotone metrics which have
monotonicity under CPTNI maps and additive noise. It is a natural
generalization of quantum monotone metrics introduced by Petz which
have monotonicity under CPTP maps. We prove the CPTNI monotone metrics
can be characterized only by operator monotone functions from few
assumptions without assuming continuities of metrics. It was shown
that CPTNI monotone metrics have some natural properties such as additivity
of direct sum (Corollary \ref{cor:direct_sum}), convexity (Corollary
\ref{cor:convexity}), monotonicity with respect to unnormalized state (Corollary
\ref{cor:monotonicity}). These properties did not appear in monotone
metrics based on CPTP maps.

\appendix

\section{Operator monotone and operator concave functions\label{sec:ope_monotone}}

A function $f:(0,\infty)\rightarrow(0,\infty)$ is said to be operator
monotone if for all positive operators $A$ and $B$, $0<A\leq B$
implies $0<f(A)\leq f(B)$. 
\begin{thm}
For a function $f:(0,\infty)\rightarrow(0,\infty)$, the following
statements are equivalent:
\begin{description}
\item [{(i)}] $f$ is operator monotone. 
\item [{(ii)}] 
$\bar{f}(C^{*}AC)\geq C^{*}\bar{f}(A)C$ for any 
positive
operator $A$, and any operator $C$ such that $\left\Vert C\right\Vert \leq1$,
where $\bar{f}:[0,\infty) \to [0,\infty)$ is a function such that $\bar{f}(x)=f(x)$ for $x>0$ and $\bar{f}(0)=\limsup_{\epsilon \searrow 0}f(\epsilon)$. 
\item [{(iii)}] $f$ is operator concave, i.e., $f(pA+(1-p)B)\geq pf(A)+(1-p)f(B)$
for any strictly positive Hermitian operators $A,B$ and any real
number $0\leq p\leq1$. 
\end{description}
\end{thm}

\begin{proof}
(ii)$\Rightarrow$(iii): Consider operators $X=\begin{pmatrix}A & 0\\
0 & B
\end{pmatrix}$, $U=\begin{pmatrix}\sqrt{p}I & \sqrt{1-p}I\\
\sqrt{1-p}I & -\sqrt{p}I
\end{pmatrix}$, $P=\begin{pmatrix}I & 0\\
0 & 0
\end{pmatrix}$. Because $\left\Vert UP \right\Vert \leq1$,
\[
\begin{pmatrix}\bar{f}\left(pA+(1-p)B\right) & 0\\
0 & \bar{f}(0)
\end{pmatrix}=\bar{f}(PUXUP)\geq PU\bar{f}(X)UP=\begin{pmatrix}p\bar{f}(A)+(1-p)\bar{f}(B) & 0\\
0 & 0
\end{pmatrix}.
\]
This proves $\bar{f}$ and $f$ are operator concave.

(iii)$\Rightarrow$(i): For operators $A$ and $B$ such that $0<A\leq A+B$,
\begin{align*}
f(A+B) & =f(p\frac{1}{p}A+(1-p)\frac{1}{1-p}B)\geq pf\left(\frac{1}{p}A\right)+(1-p)f\left(\frac{1}{1-p}B\right)\\
 & \geq pf\left(\frac{1}{p}A\right).
\end{align*}
Because every concave function is continuous, $\lim_{p\to1}pf\left(\frac{1}{p}A\right)=f(A)$.
This proves $f(A+B)\geq f(A)$. 

(i)$\Rightarrow$(ii): Without loss of generality, we can assume $0\leq C\leq I$
because any operator $C$ such that $\left\Vert C\right\Vert \leq1$
has a singular value decomposition $C=SW$ with $0\leq S\leq I$ and
an unitary operator $W$. Let $U=\begin{pmatrix}C & D\\
D & -C
\end{pmatrix}$ be an unitary operator with $D=\sqrt{I-C^{2}}$. For any real number
$\epsilon>0$, there exists $\mu>0$ such that
\[
\begin{pmatrix}CAC+2\epsilon I & 0\\
0 & \mu I
\end{pmatrix}\ge\begin{pmatrix}CAC+\epsilon I & CAD\\
DAC & DAD+\epsilon I
\end{pmatrix}=U\begin{pmatrix}A+\epsilon I & 0\\
0 & \epsilon I
\end{pmatrix}U.
\]
Then
\begin{align*}
\begin{pmatrix}f(CAC+2\epsilon I) & 0\\
0 & f(\mu I)
\end{pmatrix} & \ge U\begin{pmatrix}f(A+\epsilon I) & 0\\
0 & f(\epsilon I)
\end{pmatrix}U\ge U\begin{pmatrix}f(A+\epsilon I) & 0\\
0 & 0
\end{pmatrix}U\\
& =\begin{pmatrix}Cf(A+\epsilon I)C & Cf(A+\epsilon I)D\\
Df(A+\epsilon I)C & Df(A+\epsilon I)D
\end{pmatrix}.
\end{align*}
Therefore
\[
f^{+}(CAC)\geq Cf^{+}(A)C,
\]
where $f^{+}(x)=\limsup_{\epsilon \searrow 0} f(x+\epsilon)$ for $x \in [0,\infty)$.
Due to the proof of  (ii)$\Rightarrow$(iii), $f^+$ is operator concave and continuous at every $x>0$. 
Hence, $f^{+}(x)=f(x)$ for $x>0$. 
\end{proof}
\begin{thm}
If a function $f:(0,\infty)\rightarrow(0,\infty)$ is operator monotone,
\[
f^{\perp}(x)=x/f(x)
\]
 and 
\[
f'(x)=xf(1/x)
\]
 are also operator monotone. 
\end{thm}

\begin{proof}
Let $A,B$ be positive operators such that $0<A\leq B$. Because a
operator $C=B^{-\frac{1}{2}}A^{\frac{1}{2}}$ satisfies $\left\Vert C\right\Vert \leq1$,
\[
f(A)=f(C^{*}BC)\geq C^{*}f(B)C=A^{\frac{1}{2}}B^{-\frac{1}{2}}f(B)B^{-\frac{1}{2}}A^{\frac{1}{2}}
\]
due to the above theorem. Therefore
\[
B^{\frac{1}{2}}f(B)^{-1}B^{\frac{1}{2}}\geq A^{\frac{1}{2}}f(A)^{-1}A^{\frac{1}{2}}.
\]
This proves $f^{\perp}$ is operator monotone. Further, $f'$ is also
operator monotone because $f'(x)=1/f^{\perp}(1/x)$. 
\end{proof}

\section{Operator means\label{sec:ope_mean}}

Operator mean of strictly positive operators $A$ and $B$ with respect
to an operator monotone function is defined by
\[
Am_{f}B=A^{1/2}f\left(A^{-1/2}BA^{-1/2}\right)A^{1/2}.
\]
Later, this is extended to non-negative operators.
\begin{thm}
For strictly positive operators $A,B$ on a Hilbert space $\C^{n},$
$A\,m_{f}B=B\,m_{f'}A$.
\end{thm}

\begin{proof}
By using a singular value decomposition $A^{1/2}B^{-1/2}=S^{1/2}W$
with $0 < S$ and an unitary operator $W$,
\begin{align*}
Am_{f}B & =A^{1/2}f\left(A^{-1/2}BA^{-1/2}\right)A^{1/2}\\
 & =B^{1/2}W^{*}S^{1/2}f\left(S^{-1}\right)S^{1/2}WB^{1/2}\\
 & =B^{1/2}W^{*}f'(S)WB^{1/2}\\
 & =B^{1/2}f'(B^{-1/2}AB^{-1/2})B^{1/2}.
\end{align*}
\end{proof}
\begin{thm}
For strictly positive operators $A_{1},A_{2},B_{1},B_{2}$ on a Hilbert
space $\C^{n}$ such that $0<A_{1}\leq A_{2}$ and $0<B_{1}\leq B_{2}$,
\[
A_{1}m_{f}B_{1}\leq A_{2}m_{f}B_{2}.
\]
\end{thm}

\begin{proof}
Because $f$ and $f'$ are operator monotone,
\begin{align*}
A_{1}m_{f}B_{1} & =A_{1}^{1/2}f\left(A_{1}^{-1/2}B_{1}A_{1}^{-1/2}\right)A_{1}^{1/2}\\
 & \leq A_{1}^{1/2}f\left(A_{1}^{-1/2}B_{2}A_{1}^{-1/2}\right)A_{1}^{1/2}\\
 & =B_{2}^{1/2}f'\left(B_{2}^{-1/2}A_{1}B{}_{2}^{-1/2}\right)B_{2}^{1/2}\\
 & \leq B_{2}^{1/2}f'\left(B_{2}^{-1/2}A_{2}B{}_{2}^{-1/2}\right)B_{2}^{1/2}=B_{2}m_{f'}A_{2}.
\end{align*}
\end{proof}
Due to this theorem, operator mean of non-negative operators $A$
and $B$ can be defined by
\[
Am_{f}B=\lim_{\epsilon\ssearrow0}(A+\epsilon I)m_{f}(B+\epsilon I)
\]
because this monotonically decreases, as $\epsilon\ssearrow0$, and
the limit exists. 
\begin{thm}
For non-negative operators $A$ and $B$ on $\C^{n}$ and a linear
map $C$ from $\C^{n}$ to $\C^{m}$,
\[
C(A\,m_{f}B)C^{*}\leq(CAC^{*})\,m_{f}(CBC^{*}).
\]
\end{thm}

\begin{proof}
When $A>0$, $B>0$, and $C$ is invertible, by using a singular value
decomposition $CA^{1/2}=WS^{1/2}$ with $S>0$ and an unitary operator
$W$, 
\begin{align*}
(CAC^{*})\,m_{f}(CBC^{*}) & =(WSW^{*})\,m_{f}(CBC^{*})\\
 & =WS^{1/2}W^{*}f(WS^{-1/2}W^{*}CBC^{*}WS^{-1/2}W^{*})WS^{1/2}W^{*}\\
 & =WS^{1/2}f(S^{-1/2}W^{*}CBC^{*}WS^{-1/2})S^{1/2}W^{*}\\
 & =CA^{1/2}f(A^{-1/2}C^{-1}CBC^{*}C^{*^{-1}}A^{-1/2})A^{1/2}C^{*}\\
 & =CA^{1/2}f(A^{-1/2}BA^{-1/2})A^{1/2}C^{*}=C(A\,m_{f}B)C^{*}.
\end{align*}

When $A\geq0$, $B\geq0$, and $C$ is any linear map, $A,B,C$ can
be embedded to $\B(\C^{l})$ with $l=\max(n,m)$. Therefore we can
assume $A,B,C\in\B(\C^{l})$ without loss of generality. For any $\epsilon>0$,
by using a singular value decomposition $C=W_{C}S_{C}$ with $S_{C}\geq0$
and an unitary operator $W_{C}$, 
\begin{align*}
C(A\,m_{f}B)C^{*} & =\lim_{\epsilon\ssearrow0}W_{C}\left(S_{C}+\epsilon I\right)\left((A+\epsilon I)m_{f}(B+\epsilon I)\right)\left(S_{C}+\epsilon I\right)W_{C}^{*}\\
 & =\lim_{\epsilon\ssearrow0}W_{C}\left(\left\{ \left(S_{C}+\epsilon I\right)(A+\epsilon I)\left(S_{C}+\epsilon I\right)\right\} m_{f}\left\{ \left(S_{C}+\epsilon I\right)(B+\epsilon I)\left(S_{C}+\epsilon I\right)\right\} \right)W_{C}^{*}.
\end{align*}
Here, 
\begin{align*}
\left(S_{C}+\epsilon I\right)(A+\epsilon I)\left(S_{C}+\epsilon I\right) & \leq S_{C}AS_{C}+\delta_{A}(\epsilon)I
\end{align*}
with $\delta_{A}(\epsilon)=\left\Vert \left(S_{C}+\epsilon I\right)(A+\epsilon I)\left(S_{C}+\epsilon I\right)-S_{C}AS_{C}\right\Vert $.
Similarly, $\left(S_{C}+\epsilon I\right)(B+\epsilon I)\left(S_{C}+\epsilon I\right)\leq\delta_{B}(\epsilon)I$.
Let $\delta(\epsilon)=\max\{\delta_{A}(\epsilon),\delta_{B}(\epsilon)\}$.
Then
\begin{align*}
C(A\,m_{f}B)C^{*} & \leq\lim_{\epsilon\ssearrow0}W_{C}\left(\left\{ S_{C}AS_{C}+\delta(\epsilon)I\right\} m_{f}\left\{ S_{C}BS_{C}+\delta(\epsilon)I\right\} \right)W_{C}^{*}\\
 & =\lim_{\epsilon\ssearrow0}\left(\left\{ CAC^{*}+\delta(\epsilon)I\right\} m_{f}\left\{ CBC^{*}+\delta(\epsilon)I\right\} \right)=\left(CAC^{*}\right)\,m_{f}\left(CBC^{*}\right).
\end{align*}
\end{proof}

\end{document}